\documentclass{llncs}

\usepackage{amsfonts,amsmath,amssymb}
\usepackage{graphicx}
\usepackage{url}

\newcommand{\nd}{\noindent}

\newcommand{\ie}{i.e.}
\newcommand{\calL}{\mathcal{L}}
\newcommand{\calM}{\mathcal{M}}
\newcommand{\order}{\preceq}
\newcommand{\sorder}{\prec}

\newcommand{\siorder}{\succ}

\newcommand{\eg}{e.g.}

\newcommand{\lf}{\hat{f}}
\newcommand{\tf}{\tilde{f}}

\newcommand{\lfx}[1]{\hat{#1}}

\newcommand{\ii}[1]{\mbox{$(#1)$}}

\renewcommand{\vec}[1]{\mathbf{#1}}

\begin{document}

\title{Monotonic Mappings Invariant Linearisation of  Finite Posets}

\author{Nicol\'{a}s Madrid Labrador~\inst{1} \and Umberto Straccia~\inst{2}}

\institute{Dept. Matem\'{a}tica Aplicada, Universidad de M\'{a}laga,
Spain \and Istituto di Scienza e Tecnologie dell'Informazione (ISTI
- CNR), Pisa, Italy
}

\maketitle

\begin{abstract}

\noindent In this paper we describe a novel a procedure to build a
linear order from an arbitrary poset which \ii{i} preserves the
original ordering and \ii{ii} allows to extend monotonic and
antitonic mappings defined over the original poset to monotonic and
antitonic mappings over the new linear poset.

\end{abstract}

\section{Introduction}

\nd Sorting the entry data as an optimal queue in order to optimise resources is a problem that is resolved in some specific cases only, as for example, when the algorithm deals with monotonic mappings  and the data has a linear poset structure, as \eg~it arises in the context of top-k retrieval~\cite{LiC05,Lukasiewicz07top-kretrieval}, which originated the problem addressed here.~\footnote{In top-k retrieval one usually is interested in ranking and selecting the top-k answers to a query according to the answer's score, where the score has been computed using some monotone function such as $f\colon \mathbb{R}^{n} \to \mathbb{R}$. In this case, $f$ monotone iff  $f(\vec{x}) \leq f(\vec{y})$ whenever $\vec{x} \leq \vec{y}$, where $\vec{x} \leq \vec{y}$ holds when $x_{i} \leq y_{i}$ for all $i$.}

Nevertheless the scoring data may not be structured as a linear lattice in general. Such cases arises, \eg, when the score is an interval $[a,b]$, \eg~indicating the probability of an answer being true. In such cases, one still wants to rank the answers in order of ``highest``'' probability, but now the order is a partial order only. 


A posssible solution to this drawback is by redefining a linear ranking in the data. For example \cite{Li04ranksq} resolves this drawback by defining a linear ordering through a  \emph{``ranking"} operator, which transforms each data-tuple into a value $\alpha\in\mathbb{R}$. Another possibility is to use \emph{linear extensions} such as described in~\cite{Brightwell:91,loof:2010,patil:2003,GeneratingLEFast} in which a poset $(L,\leq)$ is extended to a linear poset $(L,\leq')$. 

However none of these \emph{linearisations methods} allow to handle appropriately monotonic mappings, in the sense that  monotonic and antitonic mappings defined over the original poset can be extended to monotonic and antitonic mappings over the new linear poset (see Section \ref{relatedWork}).


The contribution of this paper is the following. We will define a new linear poset $(\calL,\order)$ from an arbitrary  poset $(L,\leq)$ by satisfying three convenient conditions.  Firstly, each element in $L$ is assigned to an element in $\calL$; i.e there is a mapping $p\colon L\to \calL$. Secondly, the new linear lattice  $(\calL, \order)$ will be conservative with respect to the original ordering, that is if $x < y$ in $(L,\leq)$ then $p(x) \sorder p(y)$ in $(\calL,\order)$. Thirdly, if  $(L,\leq)$ is a linear poset then  $(L,\leq) \equiv (\calL,\leq)$.

This new poset $(\calL,\order)$ will be called the linearisation of $(L,\leq)$. Given a linearisation, the elements in the poset can be ranked linearly by using the value assigned in the linearisation. The improvement  with respect to other linear extensions is that we will be able to extend any monotonic mapping defined over $(L,\leq)$ to a monotonic mapping defined over $(\calL,\order)$. Hence, the rank (or ordering) provided by the linearisation respects the monotonicity of all monotonic mappings defined over $(L,\leq)$. 

In the following, we proceed as follows. In Section \ref{preli} we introduce the notation and basic definitions needed to define our linearisation procedure.  In Section \ref{linearisation} we define two linearisation procedures. In Section \ref{monotonyMappings} we show that monotonic and antitonic mappings can be extended to the linearisation by preserving monotonicity and antitonicity over these linearisations. Finally in Section \ref{relatedWork} we briefly compare the linearisation described in this paper with the linear extensions described in the literature and conclude.

\section{Preliminaries}\label{preli}

Let $L$ be a set. A binary relation $\leq$ determines an order relation in $L$ if the following properties hold for all $x,y,z\in L$: $x\leq x$ for all $x\in L$  (\emph{reflexivity}),  if $x\leq y$ and $y\leq x$ then $x=y$  (\emph{antisymmetry} and  if $x\leq y$ and $y\leq z$ then $x\leq z$  (\emph{transitivity}). Moreover, if the order relation also satisfies the property: for all $x,y\in L$ either $x\leq y$ or $y\leq y$ (\emph{linearity}) we say that $\leq$ determines a \emph{linear ordering} in $L$; otherwise, $\leq$ determines a \emph{partial ordering} in $L$.  A \emph{poset} is pair $(L,\leq)$ such that $\leq$ is an order relation in $L$. A \emph{chain} in a poset $(L,\leq)$ is a subset  $\{x_0,x_1,\dots,x_i,\dots\}$ of $L$ such that for all $i\in \mathbb{N}$: $x_i\neq x_{i+1}$, $x_i\leq x_{i+1}$
Usually, the chain $\{x_0,x_1,\dots,x_i,\dots\}$ is denoted as $x_0<x_1<\dots<x_i<\dots$. The \emph{cardinality} of a chain $C$ is usually called the \emph{length} of $C$. A chain $C$ is called \emph{maximal} if there is no chain $C'$ such that $C\subset C'$. 
Given a poset $(L,\leq)$ we extend $\leq$ to $L^{n}$ point wise, \ie~for $\vec{x}, \vec{y} \in L^{n}$, $\vec{x} \leq \vec{y}$ iff $x_{i} \leq y_{i}$, for all $i$.
Now, let $(L,\leq_L)$ and $(M,\leq_M)$ be two posets. A mapping  $f\colon L^n\to M^m$, where  $f(\vec{x})=(f_1(\vec{x}),\dots,f_m(\vec{x}))$  is called  \emph{monotonic} (resp. \emph{antitonic}) if $\vec{x} \leq_L \vec{y}$ implies  $f(\vec{x}) \leq_M f(\vec{y})$ (resp. $f(\vec{x})\geq_M f(\vec{y})$).

%


Let $\mathbb{X}$ be a set of subsets of a poset $(L,\leq)$. An order relation $\sqsubseteq$ defined over $\mathbb{X}$ determines a:
\begin{itemize}
\item \emph{H-ordering} if $X\sqsubseteq Y$ imply for all $x\in X$ there is $y\in Y$ such that $x\leq y$ (Hoare order).
\item \emph{S-ordering} if $X\sqsubseteq Y$ imply for all $y\in Y$ there is $x\in X$ such that $x\leq y$ (Smyth order).
\item \emph{EM-ordering} if $\sqsubseteq$ determines a H-ordering and a S-ordering (Egli-Milner order).
\end{itemize}

\section{The linearisation of a finite poset}\label{linearisation}

The aim of this section is to define, from a \emph{finite} poset $(L,\leq)$, a linear poset $(\calL, \order)$ satisfying the properties described in the introduction; namely:
\begin{itemize}
\item If $(L,\leq)$ is a linear poset then $(L,\leq) = (\calL,\preceq)$
\item  There is a map $p:L\to \calL$ such that $\forall \ x,y\in L$ if $x<y$ then $p(x)< p(y)$.
\end{itemize}

\nd To start with, we begin by defining a set of disjoint sets in order to define later an equivalence relation over $L$. We proceed inductively as follows (recall that $L$ is finite):
\begin{align*}
& L_0=\{x\in L \mid \nexists y\in L \textit{ such that } x<y  \} \\
&L_1=\{x\in L  \mid x\notin L_0 , \nexists y\in L\setminus L_0 \textit{ such that } x<y \}\\
& \quad \vdots \\
&L_i=\{x\in L \mid x\notin L_0\cup\cdots\cup L_{i-1}, \\
&  \qquad \qquad \nexists y\in L\setminus (L_0\cup\cdots\cup L_{i-1}) \textit{ such that } x<y  \}\\
& \quad \vdots 
\end{align*}
Each $L_i$ is called the $i$-level of $L$. The underlying idea under the levels is to take the values of $L$ in  decreasing order from the top elements. 
The following lemma can be shown:

\begin{lemma}\label{lemma1}
Let $(L,\leq)$ be a poset.
\begin{enumerate}
\item\label{lemma1:1}  If $x\in L$ belongs to the level $L_k$ then $x\notin L_i$ for all $i\neq k$.
\item \label{lemma1:2} If $x< y$, $x\in L_i$ and $y\in L_j$ then $i>j$. 
\item \label{lemma1:3} If $x$ and $y$ belong to the same level then $x$ and $y$ are incomparable.
\end{enumerate}
\end{lemma}


\nd The levels also satisfy the following property: 
 
\par\begin{proposition}\label{prop1}
Let $(L,\leq)$ be a finite poset. If $L_0\cup\cdots\cup L_{i-1}\neq L$ then $L_i\neq \emptyset$.
\end{proposition}
\begin{proof} Suppose $L_0\cup\cdots\cup L_{i-1}\neq L$, then there is an element $x_0\in L\setminus (L_0\cup\cdots\cup L_{i-1})$. If $x_0\notin L_i$ then there exists $x_1\in L\setminus (L_0\cup\cdots\cup L_{i-1})$ such that $x_0<x_1$. If $x_1\notin L_i$ then there exists $x_2\in  L\setminus (L_0\cup\cdots\cup L_{i-1})$ such that $x_0<x_1<x_2$. This construction cannot be infinite since $L$ is finite and therefore there exists $k\in\mathbb{N}$ such that $x_k\in L_i$.\qed
\end{proof}

\nd Note that Proposition~\ref{prop1} does not hold if the poset is not finite. For example, $[0,1]$ with the usual order does not satisfy Proposition \ref{prop1}, since $L_1=\emptyset$ and $L_0=\{1\}\neq [0,1]$.


\begin{proposition}\label{prop2}
If $L_i=\emptyset$ then $L_{i+1}=\emptyset$ as well.
\end{proposition}
\begin{proof}
As $L_i=\emptyset$ then:
\begin{eqnarray*}
L_{i+1} & = & \{x\in L \mid x\notin L_0\cup\cdots\cup L_{i}, \ \nexists y\in L\setminus (L_0\cup\cdots\cup L_{i}) \textit{ such that } x<y  \} \\
&= &\{x\in L \mid x\notin L_0\cup\cdots\cup L_{i-1},  \ \nexists y\in L\setminus (L_0\cup\cdots\cup L_{i-1}) \textit{ such that } x<y  \} \\ 
&=& L_i  \ .
\end{eqnarray*} \qed
\end{proof}

%

\begin{corollary}\label{cor1}
Let $(L,\leq)$ be a finite poset. Each $x\in L$ belongs to one, and only one, level $L_i$ of $L$.
\end{corollary}


\nd Corollary \ref{cor1} shows us that we can define an equivalence relation over a finite poset $(L,\leq)$ by identifying the elements which belong to the same level. That is, let $x$ and $y$ be two elements in $L$, we say that $x\sim y$ if and only if $x$ belongs to the same level of $y$. In other words, $x\sim y $ if and only if there exists $i\in \mathbb{N}$ such that $x,y\in L_i$. It is easy to proof that the relation $\sim$ is actually an equivalence relation since each element $x\in L$ belongs to one, and only one, level of $L$. Hence, we can consider the set $L/_\sim$, where the elements are the equivalence classes of $\sim$. In this way, $[x]$ denotes the equivalence class of $x\in L$, i.e the level which contains $x$. We also write $[L_{i}]$ in place of $[x_{i}]$ for any $x_{i} \in L_{i}$. 
Over $L/_\sim$  we can define also the following order relation:
$$
[x]\leq_{L/_\sim}[y] \quad \iff \quad x\in L_i , y\in L_j \textit{ and } j\leq i \ .
$$
It easy to proof that the above relation $\leq_{L/_\sim}$ is effectively an order relation. In fact, $\leq_{L/_\sim}$ defines a linear ordering over $L/_\sim$.  At this point, we are able to define the linearisation of a finite poset.

\begin{definition}
Let  $(L,\leq)$ be a finite poset. The linear poset $(L/_\sim,\leq_{L/_\sim})$ defined above is called de \emph{linearisation} of $(L,\leq)$.
\end{definition}

\nd To the ease of representation, we denote by $(\calL,\order)$ the linearisation of $(L,\leq)$ and have by definition $[L_{i}] \order [L_{j}]$ iff $j \leq i$. In the following, we illustrate some linearisation examples.

 \begin{example}\label{posetOne}
Consider the poset represented by the following graph: 
\begin{center}
\includegraphics[scale=0.175]{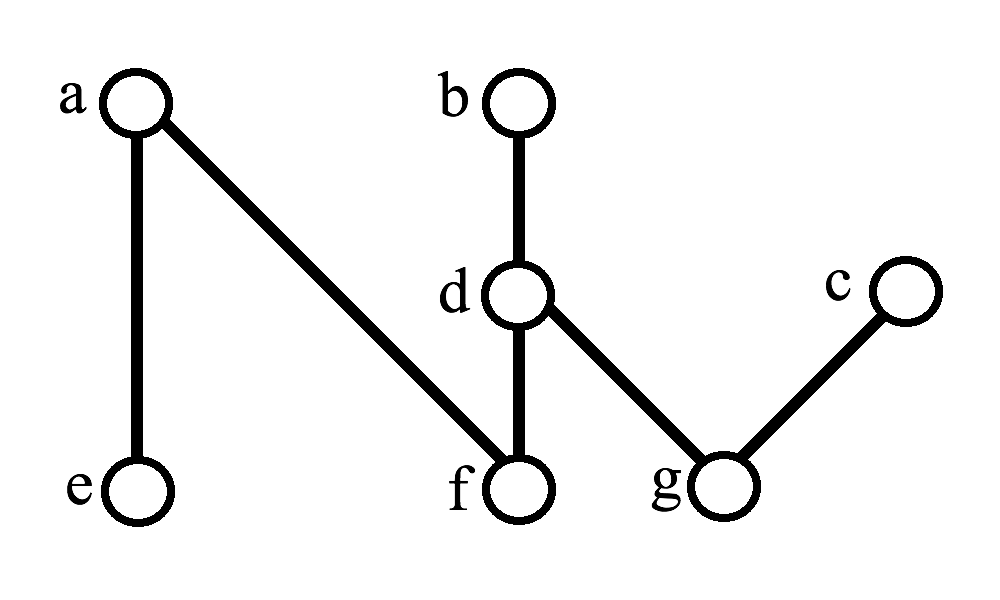}
\end{center}
The levels of this poset are $L_0=\{a,b,c\}$, $L_1=\{d,e\}$, $L_2=\{f,g\}$ and $L_i=\emptyset$ for $i\geq 3$. Then the linearisation of this poset is the chain $[f] \sorder [d] \sorder[a]$. \qed
\end{example}


 \begin{example}\label{exabc}
Consider the poset represented by the following graph: 
\begin{center}
\includegraphics[scale=0.15]{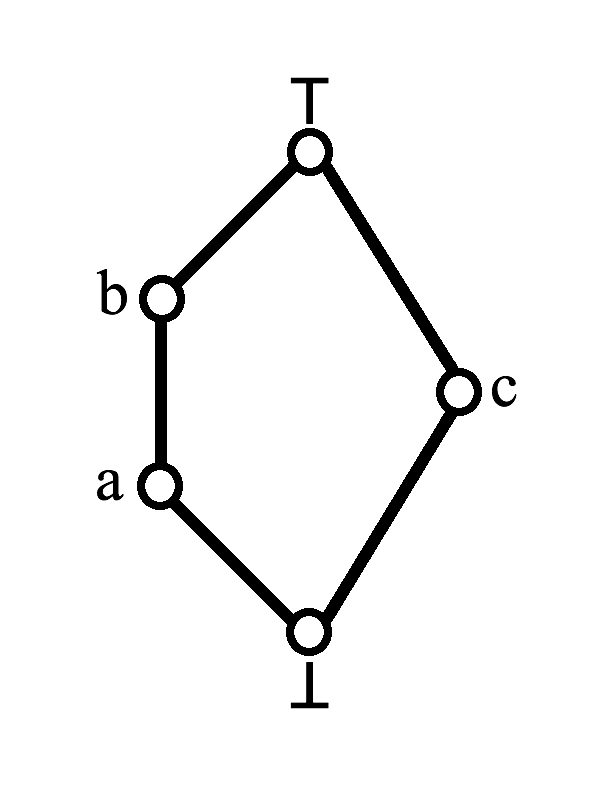}
\end{center}
Then the linearisation of this lattice is the chain $[\bot]\sorder [a]  \sorder ([b]=[c])  \sorder [\top]$. 
\qed
\end{example}

\nd Our linearisation satisfies the following properties. 

\begin{proposition}\label{prop4}
Let $(L,\leq)$ be a finite poset and let $(\calL, \order)$ be its linearisation. Then the following properties hold: 
\begin{enumerate}
\item \label{1} If $(L,\leq)$ is a linear poset then $(\calL,\order) \equiv (L,\leq)$.
\item \label{2} If $x\leq y$ then $[x]\order [y]$. More specifically: if $x < y$ then $[x] \sorder [y]$.
\item \label{3} If $[x]\order [y]$, then for all $z\in[x]$ there exists $z'\in [y]$ such that $z\leq z'$. More specifically: if $[x] \sorder [y]$, then for all $z\in[x]$ there exists $z'\in[y]$ such that $z < z'$.
\end{enumerate}
\end{proposition}

\begin{proof}
We begin with the proof of item (\ref{1}). If $\leq$ determines a linear ordering in $L$ then $(L,\leq)$ has a chain structure, i.e $(L,\leq)\equiv \bot < x_{1} < \dots < x_{n} =\top$ for some $n\in\mathbb{N}$. It is easy to check that $(\calL,\order)$ is $[\bot] \sorder [x_{1}] \sorder \dots \sorder [x_{n}] = [\top]$, with $[\bot] = \{\bot\}, [\top] = \{\top\}$ and $[x_{i}] = \{x_{i}\}$ for all $1 \leq i \leq n$. So, $(L,\leq) \equiv (\calL,\order)$.

Item (\ref{2}) is a direct consequence of Lemma \ref{lemma1}. 
 
Concerning item (\ref{3}), suppose $[x] \sorder [y]$ (the case $[x]=[y]$ holds trivially). Let $L_i$ be the level associated with the class of $[x]$. Consider $z\in L_i$. Then, as $z\notin L_{i-1}$, there exists $z_1\in  L\setminus (L_0\cup\cdots\cup L_{i-2})$ such that $z<z_1$. Moreover, as $z\in L_i$ there does not exist $t\in L\setminus (L_0\cup\cdots\cup L_{i-1})$ such that $z<t$. In other words, $z_1$ belongs to $L_{i-1}$. If we continue with this procedure, we can build a chain as follows:
$$
z<z_1<\cdots<z_j<\cdots<z_i=\top
$$
where each $z_j\in L_{i-j}$. To end the proof, we only have to select the element $z_k$ in the chain such that $z_k\in[y]$.\qed
\end{proof} 

\begin{remark}
Firstly we have just obtained the aim described in the introduction by taking the projection $p:L\to\calL$ which assigns to each $x\in L$ the level $[x]\in \calL$. Secondly, the item \ref{3} implies that $\preceq$ defines a $H$-ordering among the levels of a finite poset. \qed
\end{remark}

\begin{proposition}\label{cardinality}
Let $(L,\leq)$ be a finite poset, then the number of elements of $(\calL,\preceq)$ coincides with the maximum among the cardinalities of chains in $(L,\leq)$.
\end{proposition}

\begin{proof}
As $(L,\leq)$ is finite, every chain has a finite length. Let $C$ be a chain $x_1<\cdots<x_n$ with maximal cardinality. Then, as the elements in the same level have to be incomparable, necessarily there exists at least $n$ levels; \ie, $\calL$ has at least the same number of elements than the cardinality of $C$.

On the other hand, if the number of elements of $\calL$ is $n$, $\calL$ coincides with the chain 
$[x_{n-1}] \sorder [x_{n-2}] \sorder \cdots \sorder [x_1] \sorder [x_0]$ for arbitrary $x_{i} \in L_{i}$. 

Take an element $y_1\in L_{n-1}$. By Proposition \ref{prop4} item \ref{3}, there is an element  $y_2\in L_{n-2}$ such that $y_2>y_1$. Similarly, by using the same result, there is an element $y_3\in L_{n-3}$ such that $y_3>y_2>y_1$. By following this procedure we can build a chain $y_n>y_{n-1}>\cdots>y_2>y_1$ with cardinality that is then the number of elements of $\calL$. Therefore the maximun among the cardinalities of chains of $(L,\leq)$ is greater or equal than the number of elements of $\calL$.

Unifying both inequalities, we conclude than the number of elements of $(\calL,\order)$ coincides with the maximum among the cardinalities of chains in  $(L,\leq)$.\qed
\end{proof}

\subsection{The Dual-linearisation}\label{duallinearisation}

\nd Clearly, the levels defined in the above section can be defined in a dual way.  More exactly, we can define the dual-levels by taking incomparable elements in increasing order instead of decreasing order as in the procedure describe in the previous section. Therefore a dual linearisation procedure can be defined. In this section we will describe the dual-linearisation and we will provide results (or dual-results). 

Let $(L,\leq)$ be a poset, we define the dual-levels as follows:
\begin{align*}
&L_0^*=\{x\in L \mid \nexists y\in L \textit{ such that } x>y \}\\
&L_1^*=\{x\in L \mid x\notin L_0^* , \nexists y\in L\setminus L_0^* \textit{ such that } x>y \}\\
&\quad \vdots\\
& L_i^*=\{x\in L \mid x\notin L_0^*\cup\cdots\cup L_{i-1}^*, \\
&  \qquad \qquad\ \ \nexists y\in L/(L_0^*\cup\cdots\cup L_{i-1}^*) \textit{ such that } x>y  \}\\
&\quad \vdots \ 
\end{align*}

\nd Each $L_i^*$ is called the $i$-dual-level of $L$. Lemma \ref{lemma1} is rewritten for dual-levels as follows: 

\begin{lemma}\label{lemma2}
Let $(L,\leq)$ be a poset.
\begin{enumerate}
\item\label{lemma2:1}  If $x\in L$ belongs to the level $L_k$ then $x\notin L_i$ for all $i\neq k$.
\item \label{lemma2:2} If $x< y$, $x\in L_i$ and $y\in L_j$ then $i<j$. 
\item \label{lemma2:3} If $x$ and $y$ belong to the same level then $x$ and $y$ are incomparable.
\end{enumerate}
\end{lemma}


\nd Propositions~\ref{prop1} and \ref{prop2} and  Corollary~\ref{cor1} could be rewriting (and prooved) substituting directly  the levels $L_i$ by the dual levels $L^*_i$. This feature allows us to define the equivalence relation $\sim_* $ as follows: $x\sim_* y$ if and only if $x$ and $y$ belong to the same dual-level. Hence, we can consider the quotient set $L/_{\sim_{*}}$ with the ordering:
$$
[x]\order_* [y] \ \textit{ if and only if } \ x\in L^*_i \ , \ y\in L^*_j \ \textit{and}\ i\leq j.
$$
As  for the order relation $\order$ defined over $\calL$, $\order_*$ defines a linear ordering in $L/_{\sim_{*}}$.

\begin{definition}
Let  $(L,\leq)$ be a finite poset. The linear poset $(L/_{\sim_{*}},\order_*)$ defined above is called de \emph{dual-linearisation} of $(L,\leq)$.
\end{definition}
  
\nd We use the notation $(\calL^*,\order_*)$ to denote the dual-linearisation of  $(L,\leq)$, and have by definition $[L_{i}] \order_{*} [L_{j}]$ iff $i \leq j$,
where we write $[L^{*}_{i}]$ in place of $[x_{i}]$ for any $x_{i} \in L^{*}_{i}$.
Obviously, the dual-linearisation and the linearisation may not coincide.
 
 \begin{example}\label{exampleABCLatticeDual}
 Consider the poset given in Example \ref{exabc}.  Then dual-linearisation is the chain:
 \[
 [\bot] \order_* ([c] = [a]) \order_* [b] \order_* [\top] \ .
\]
 Note that in the linearisation $[a]\neq[c]$ whereas in the dual-linearisation $[a]=[c]$. \qed
 \end{example}
 
\nd Proposition $\ref{prop4}$ can be rewriting for the dual-linearisation by  changing slightly item $3$. 

\begin{proposition}\label{prop4dual} 
Let $(L,\leq)$ be a finite poset and let \ $(\calL^*,\order_*)$ be its dual-linearisation. Then the following properties hold: 
\begin{enumerate}
\item \label{b1} If $(L,\leq)$ is a linear poset then $(\calL^*,\order_*) \equiv (L,\leq_L)$.
\item \label{b2} If $x\leq y$ then $[x] \order_*[y]$. More specifically: if $x < y$ then $[x] \sorder_* [y]$.
\item \label{b3} If $[x]\order_* [y]$, then for all $z\in[y]$ there exists $z' \in [x]$ such that $z'\leq y$. More specifically: if $[x] \sorder_*[y]$, then for all $z\in[y]$ there exists $z'\in[x]$ such that $z' < z$.
\end{enumerate}
\end{proposition}


\begin{remark}
Note that item \ref{3} in  Proposition \ref{prop4dual} shows the most important difference between the levels and the dual-levels. That is, the set of levels is $H$-ordered by $\order$ whereas the set of dual-levels is $S$-ordered by $\order_*$.  \qed
\end{remark}

\nd As in the linearisation, the number of element in the dual linearisation is determined by the length of the maximal chains in $(L,\leq)$.

\begin{proposition}\label{propdual}
Let $(L,\leq)$ be a finite poset, then the number of elements of $(\calL^*,\order_*)$ coincides with the maximum among the cardinalities of chains in $(L,\leq)$.
\end{proposition}

\subsection{Relationships between linearisation and dual linearisation}\label{relationship}

\nd The aim of this section is to determine a condition in $(L,\leq)$, which characterises the equivalence $(\calL,\order)\equiv (\calL^*,\order_*)$. Notice that if that equivalence holds then $\order$ (or equivalently  $\order_*$) determines an $EM$-ordering in the set of levels and in the set of dual-levels of $(L,\leq)$. It is important to note, as we show in Example \ref{exampleABCLatticeDual}, that usually  $(\calL,\order)$ and $(\calL^*,\order_{*})$ are not equivalent. However, we can ensure that  $\calL$ and $\calL^*$ 
have the same number of elements, thanks to Propositions \ref{cardinality} and \ref{propdual}.

\begin{corollary}
Let $(L,\leq)$ be a finite poset. Then $|\calL|=|\calL^*|$.
\end{corollary}


\nd The condition which characterises the equivalence $(\calL,\order)\equiv(\calL^*,\order_{*})$ is given in the following definition:

\begin{definition}
A finite poset satisfies the \emph{equal length chain condition} (\emph{ELCC}) if and only if all maximal chains have the same length. 
\end{definition} 


\begin{example}
The poset given in example \ref{exabc} does not satisfy the ELCC since  $\bot < a < b <\top$ and $\bot < c < \top$ are two maximal chains with different length; $4$ and $3$ respectively. \qed
\end{example}

\begin{example}
The following two posets satisfy the ELCC:
$$
\begin{array}{ccc}
\includegraphics[scale=0.2]{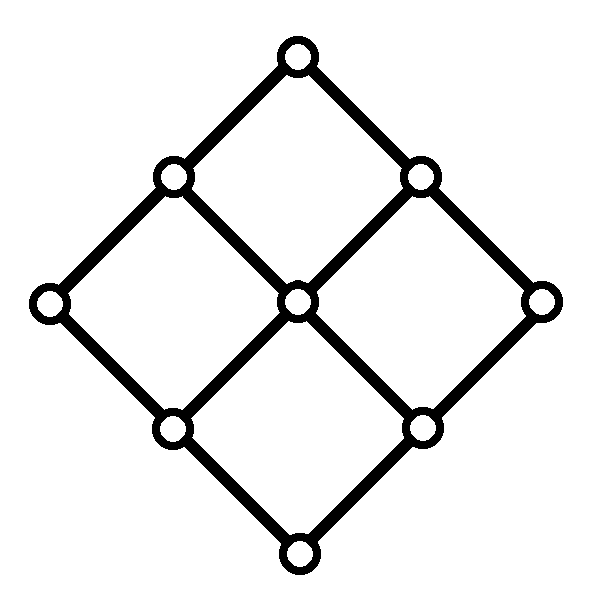} & \qquad \ \qquad & \includegraphics[scale=0.2]{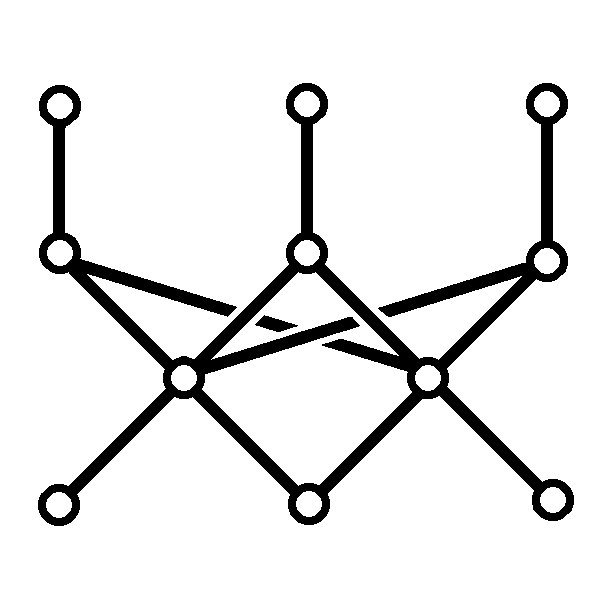}
\end{array}
$$
In the left poset every maximal chain has a length $5$, whereas in the right poset every maximal chain has a length of $4$. \qed
\end{example}


\begin{proposition}
 Let $(L,\leq)$ be a finite poset.  $(L,\leq)$ satisfies the ELCC if and only if $(\calL,\order) \equiv (\calL^*,\order_*)$; i.e for all $x\in L$:\begin{itemize}
\item the equivalence class of $x$ in $\calL$ and $\calL^*$ coincides;
\item $[x] \order [y]$ if and only if $[x] \order_{*}[y]$.
\end{itemize}
\end{proposition}

\begin{proof}
Let $n\in\mathbb{N}$ be the length of each maximal chain in $(L,\leq)$. Let us prove that each level $L_i$ coincides with the dual level $L^*_{n-i}$ for $i=1,\cdots,n$. 

Consider $x\in L$ and  suppose that $x\in L_i$ and $x\in\calL^*_j$. By the Proposition~\ref{prop4} item \ref{3}, we can build the chain:
$$
x< y_1< y_2<\cdots<\top \ ,
$$ 
where each $y_k\in L_{i-k}$. 

Dually, by the Proposition~\ref{prop4dual} item \ref{3}, we can build the chain
$$
\bot<\cdots<z_2<z_1<x \ ,
$$
where each $z_k\in L^*_{j-k}$. Unifying both chains we obtain the maximal chain
$$
\bot<\cdots<z_2<z_1<x< y_1< y_2<\cdots<\top \ ,
$$
whose length is $j+i$. As $(L,\leq)$ satisfies the ELCC, $j+i=n$ and therefore $j=n-i$. Therefore $L^*_{n-i}=L_i$.

Finally, $[L_i]\order [L_j]$\ iff \ $i\geq j$ \ iff \ $n-j\geq n-i$ \ iff\ $[L^*_{n-i}]\order_{*} [L^*_{n-j}]$ \ iff\  $[L_i]\order_{*} [L_j]$.

Let us prove the other implication. That is, $(\calL,\order)\equiv(\calL^*,\order_{*})$ implies that $(L,\leq)$ satisfies the ELCC. 

We proceed by induction over the number of elements in $L$. For $i=1$ the proof is trivial since every poset with only one point satisfies the ELCC and its linearisation and dual-linearisation coincide. Suppose that the number of elements of $L$ is $n > 1$ and that the result is true for every poset with a number of elements less than $n$. Consider now $L\setminus L_0=L\setminus\{x\in L \colon \nexists y\in L \textit{ such that } x>y \}$. As $(\calL,\order)\equiv(\calL^*,\order_{*})$ then necessarily $((L\setminus L_0),\leq)\equiv((L\setminus L_0)^*,\order_{*})$. Therefore, by induction hypotesis, $L\setminus L_0$ satisfies the ELCC. Let $k$ be the length of every maximal chain in $L\setminus L_0$. Let us show that the length of every maximal chain in $L$ is $k+1$. Let  $C$ be a maximal chain of $(L,\leq)$ with length $k'$. By Proposition~\ref{prop4} item \ref{3}, the maximal element in $z\in C$ has to belong to $L_0$. Moreover $C\setminus \{z\}$ is a maximal chain of $L\setminus L_0$. Then by using the  induction hypotesis, the length of $C\setminus \{z\}$ is $k$ and therefore the length of $C$ is $k+1$. \qed
\end{proof}

\section{Extending monotonic mappings  over the linearisation}\label{monotonyMappings}

We take up again the issue to extend monotonic mappings over a linearized poset $(\calL,\order)$ or dually over $(\calL^*,\order_{*})$. The main problem of extending monotonic mappings in a linearisation procedure is that it \emph{cannot} been done by homomorphic extension if the properties described in the introduction are satisfied. 
The following lemma shows this (unexpected) negative feature for a special family of posets: the lattices.~\footnote{A lattice~\cite{GeneralLatticeTheory} is a poset $(L,\leq)$ such that for every pair of elements $a,b\in L$ the operators supremum ($\sup$) and infimum ($\inf$) are defined.}
 
\begin{proposition}\label{teo1}
Let $(L_{1},\leq_{1})$ be a lattice such that $\leq_{1}$ is not a linear ordering. Then there does not exist a linear lattice $(L_{2},\leq_{2})$ together with a map 
$p\colon L_{1}\to L_{2}$ such that  
$\forall x,y\in L_{1}$:
\begin{enumerate}
\item $x <_{1} y$ implies $p(x) <_{2} p(y)$; 
\item for every monotonic mapping  $f\colon L_{1}\to L_{1}$, the mapping   $\lf\colon L_{2}\to L_{2}$ defined by  $\lf(p(x))=p(f(x))$ is well defined and monotonic.
\end{enumerate}
\end{proposition}

\begin{proof} As $(L_{1},\leq_{1})$ is not a linear lattice, there exists two incomparable elements $a,b\in L_{1}$. Without loss of generality, let us suppose that $p(a) \leq_{2} p(b)$. If $p(a) = p(b)$ then the extension in $(L_{2},\leq_{2})$ of the monotonic mapping $f(x)=\sup\{x,a\}$ is not well defined: in fact 
\begin{eqnarray*}
p(a) & = & p(\sup\{a,a\})  =  p(f(a)) \\
& = & \lf(p(a))  =  \lf(p(b)) \\
& = & p(f(b))  =  p(f(\sup\{b,a\})) \\
& >_{2} & p(a) \ .
\end{eqnarray*}

\nd If $p(a) <_{2} p(b)$ then the extension in $(L_{2},\leq_{2})$  of any monotonic mapping such that $f(a)=b$ and $f(b)=a$ cannot be monotonic. In fact,  $p(a) <_{2} p(b)$, but 
\[
\lf(p(a)) = p(f(a)) = p(b) >_{2} p(a) = p(f(b))= \lf(p(b)) \ . 
\]
\qed
\end{proof}


\nd  However,  on the positive side, if $(L,\leq_L)$ and $(M,\leq_M)$) are \emph{finite} posets and $f\colon L^n \to M$ is a mapping, we can extend $f$ over our linearisation $(\calL,\order_{\calL})$ and $(\calM,\order_{\calM})$, respectively, by the following two ways: $\lf \colon \calL^{n} \to \calM$ and $\tf \colon \calL^{n} \to \calM$, where
$$
\lf([x_1],\cdots,[x_n])=\max\{[f(y_1,\cdots,y_n)] \mid y_1\in[x_1],\cdots, y_n\in[x_n]\}
$$
and
$$
\tf([x_1],\cdots,[x_n])=\min\{[f(y_1,\cdots,y_n)] \mid y_1\in[x_1],\cdots, y_n\in[x_n]\} \ .
$$
Obviously $\lf$ and $\tf$ are well-defined. Note the following abuse of notation in the above formulae: we use $[\cdot]$ to  denote  equivalence classes in $(\calL,\order_{\calL})$ as well as in $(\calM,\order_{\calM})$. We keep this abuse of notation in the rest of the section since is clear when $[\cdot]$ is either an element of $(\calL,\order_{\calL})$ or an element of $(\calM, \order_{\calM})$. 


More general mappings $f\colon L^n\to M^m$ can be extended to a mapping from $\calL^n$ to $\calM^m$ as well using  the both ideas above. 
In this case $f$ can be seen as $m$ mappings $f_i\colon L^n\to M$ and $\lf$ and $\tf$ are defined by:
\begin{eqnarray*}
\lf([x_1],\cdots,[x_n]) & = & (\lf_1([x_1],\cdots,[x_n]),\cdots,\lf_m([x_1],\cdots,[x_n])) \\\\
\tf([x_1],\cdots,[x_n]) & = & (\tf_1([x_1],\cdots,[x_n]),\cdots,\tf_m([x_1],\cdots,[x_n])) \ .
\end{eqnarray*}

\nd  The monotonicity (resp. antitonicity) of $\lf$ and $\tf$ is determined by the monotony (resp. antitonicity) of each $\lf_i$ and $\tf_i$, respectively. 
As a consequnece,  we may restrict our study on mappings to the former case in which the codomain is a poset $(M,\leq_M)$. The following results holds:

\begin{proposition}
Let $(L,\leq_L)$ and $(M,\leq_M)$ be two finite posets. Then, for every mapping  $f\colon L^n\to M$, the following inequalities hold for all $x_1,\cdots,x_n\in L$:
\[
\tf([x_1],\cdots,[x_n]) \order_{\calM} [f(x_1,\cdots,x_n)] \order_{\calM} \lf([x_1],\cdots,[x_n]) \ .
\]
\end{proposition}

\nd The following propositions show us that in the linearisation $\lf$ keeps the monotonicity of $f$ whereas $\tf$ keeps the antitonicity  of $f$:

\begin{proposition}\label{monotonyOver}
If $f:L^n\to M$ is monotonic then $\lf:\calL^n\to \calM$ is monotonic as well.
\end{proposition}

\begin{proof} 
Assume $[x_i] \order_{L} [y_i]$ for $i=1,\cdots,n$. From Proposition~\ref{prop4} item \ref{3}, for any tuple $(v_1,\cdots,v_n)\in[x_1]\times\cdots\times[x_n]$ there exists a tuple $(t_1,\cdots,t_n)\in[y_1]\times\cdots\times[y_n]$ such that  $v_i\leq t_i$ for all $i=1,\cdots,n$. Then, by using the monotonicity of $f$ and Proposition \ref{prop4} item $\ref{2}$ we have the inequality   $[f(v_1,\cdots,v_n)] \preceq_{\calM} [f(t_1,\cdots,t_n)]$. Therefore:
\begin{eqnarray*}
\lf([x_1],\cdots,[x_n]) & = & \max\{[f(z_1,\cdots,z_n)] \mid z_1\in[x_1],\cdots, z_n\in[x_n]\} \\
& \order_{M} & \max\{[f(t_1,\cdots,t_n)] \mid t_1\in[y_1],\cdots, t_n\in[y_n]\} \\
& \order_{M} & \lf([y_1],\cdots,[y_n]) \ .
\end{eqnarray*}
\qed
\end{proof}

\nd In a similar way we may prove that:
\begin{proposition}\label{antitonicityUnder}
If $f:L^n\to M$ is antitonic then $\tf:\calL^n\to \calM $ is antitonic as well.
\end{proposition}


\nd It is important to note that $\lf$ and $\tf$ do not necessarily extend antitonic and monotonic mappings, respectively. The following example shows that feature:

\begin{example}
Over the poset  given in Example \ref{exabc} we consider the following monotonic mapping  $f$ and antitonic mapping  $g$:
$$
\begin{array}{ccc}
f(x)= \left \{ \begin{array}{ccc}
\bot & \text{if} & x=\bot \\
c& \text{if} & x=c\\
\top & \text{if} & x=\top,a,b
\end{array}
\right.
& \quad \ \quad &
g(x)= \left \{ \begin{array}{ccc}
\top & \text{if} & x=\bot \\
c& \text{if} & x=c\\
\bot & \text{if} & x=\top,a,b
\end{array}
\right .
\end{array}
$$
It is easy to check that, despite being $f$ monotonic, the mapping 
$$
\tf([x])= \left \{ \begin{array}{ccc}
[\bot] & \text{if} & [x]=[\bot] \\
\text{[c]}& \text{if} & [x]=[c]\\
\top &\text{if} & [x]=[a]
\end{array}
\right.
$$
is not monotonic ($[a] \order [c]$ and $\tf([a]) = [\top] \siorder \tf([c]) = [c]$). On the other hand, it is also easy to check that, despite being $g$ antitonic, the mapping :
$$
\lfx{g}([x])= \left \{ \begin{array}{ccc}
[\top] & \text{if} & [x]=[\bot] \\
\text{[c]}& \text{if} & [x]=[c]\\
\bot & \text{if} & [x]=[a]
\end{array}
\right.
$$
is not antitonic ($[a]\leq[c]$ and $\overline{g}([a])=[\bot]<\overline{g}([c])=[c]$). \qed
\end{example}

\nd Obviously, we can extend mappings through $\lf$ and $\tf$ in the dual-linearisation as well, which are defined for $(\calL^*,\order_{*})$ similarly as  for $(\calL,\order)$. However, the behaviour of these extensions is dual. More precisely, the extension $\lf$ keeps the antitonicity, while $\tf$ keeps the  monotonicity in the dual-linearisation.

\begin{proposition}\label{propmapping }\label{monotonyOverDual}
If $f:L^n\to M$ is antitonic then $\lf:(\calL^*)^n\to \calM$ is antitonic as well.
\end{proposition}


\begin{proposition}\label{propmappingB }\label{antitonicityUnderDual}
If $f:L^n\to L$ is monotonic then $\tf:(\calL^*)^n\to \calM$ is monotonic as well.
\end{proposition}

%

\nd The proofs of both results are essentially the same as for Proposition~\ref{monotonyOver}. 

By using Proposition~\ref{monotonyOverDual}, Proposition~\ref{antitonicityUnderDual} and the results provided in Section \ref{relationship}, 
we prove that:

%

\begin{corollary}
Let $(L,\leq_L)$ and $(M,\leq_M)$ be two finite posets such that $(L,\leq_L)$ satisfies the ELCC. Then if $f\colon L^n\to M$ is monotonic (resp. antitonic) then $\lf$ and $\tf$ are monotonic (resp. antitonic). \footnote{Independently if $\lf$ and $\tf$ is defined over either the linearisation or the dual-linearisation.} 
\end{corollary}

%

\section{Related Work \& Conclusions}\label{relatedWork}

Linear extensions have usually be seen as the extension of a poset $(L,\leq)$ to a linear poset $(L,\leq')$, where the set $L$ is not changed~\cite{Brightwell:91,loof:2010,patil:2003,GeneratingLEFast}.  The main  drawback of this kind of linearisation is that all of them prevents to extend monotonic mappings. This feature can be derived from Proposition~\ref{teo1}, as the unique natural way to extend mappings to a linear extension is by homomorphic extension. Besides that, another drawback of linear extensions relates to the number of different possible linear extensions to look at. In \cite{Brightwell:91}, the authors resolved the problem of counting the number of linear extensions by showing that the problem is \emph{\#P-complete}, which may prevent it to be used in practice.


%
%

We instead have shown that it is possible to linearise posets, under the condition that they are finite, such that the monotonicity (resp. antitonicity) of arbitrary mappings is preserved.



%
%
%

\renewcommand{\baselinestretch}{1}
\setlength{\textwidth}{12.2cm}
\setlength{\textheight}{19.3cm}

\end{document}